\documentclass[11pt]{article}
\usepackage{fullpage}
\usepackage{amssymb,latexsym,amsmath}     
\usepackage [latin1]{inputenc}
\usepackage{graphicx}
\usepackage{authblk}
\usepackage{dsfont}
\usepackage{amsthm}
\usepackage{color}
\usepackage{float,todonotes}
\usepackage{amsmath}
\providecommand{\DP}{\ensuremath{\mathsf{DP}}}
\usepackage{comment}
\usepackage{tcolorbox}
\usepackage[colorlinks=true,linkcolor=blue, citecolor=blue,urlcolor=blue]{hyperref}
\newcommand{\tododiya}[1]
{\todo[linecolor=orange,backgroundcolor=orange!50,bordercolor=pink]{#1}}  
\usepackage{algorithm}

\usepackage{algpseudocode}
\theoremstyle{definition}
\newtheorem{definition}{Definition}[section]
\theoremstyle{plain}
\newtheorem{theorem}{Theorem}[section]
\newtheorem{lemma}{Lemma}[section]
\newtheorem{claim}{Claim}[section]
\newtheorem{observation}{Observation}

\theoremstyle{remark}

\newtheorem*{observation*}{Observation}

\begin{document}

\title{Exploiting Graph Structure for Near-Optimal Broadcasting}

\author[1]{Rudranarayan Kar}
\author[2]{Praneet Kumar Patra}
\author[1]{Diya Roy}
\author[1]{Abhishek Sahu}

\affil[1]{NISER Bhubaneswar, An OCC of Homi Bhabha National Institute, Mumbai}
\affil[2]{IISER Pune}

\date{}
\maketitle

\begin{abstract}
   Telephone broadcasting is a well-studied model for spreading information in a network as fast as possible. Here, the network is modeled as a connected graph $G(V,E)$, and a piece of information starts at a source vertex $s$. The only rule is that, in one time step, an informed vertex can inform exactly one of its uninformed neighbors. Decision version asks, given a source $s$ and a positive integer $t$, whether the information can be spread to every vertex of $G$ within $t$ time steps; we call this the broadcast time and denote it $b(G,s)$. A related version does not fix the source in advance; here we look at the worst-case broadcast time over all possible sources, $b(G) = \max_{u \in V} b(G,u)$. Both versions of the problem are known to be NP-hard. Also mark that the broadcast time of any graph on $n$ vertices must be at least $\lceil \log_2 n \rceil$. Recently, Fomin \textit{et al.} \cite{fomin2023parameterized} showed that this problem is fixed-parameter tractable (FPT) with respect to certain structural parameters of the graph. In this paper, our goal is : instead of insisting on an optimal broadcast schedule, we look for algorithms that find a valid (possibly non-optimal) broadcast schedule much faster. We study the Broadcasting problem from a parameterized and exact-algorithms perspective, and improve two results of Fomin \textit{et al.}~\cite{fomin2023parameterized}. First, we improve their exact algorithm running in time $\mathcal{O}^*(3^n)$ to $\mathcal{O}^*(3-f(x))^n$ for a +$x$ approximation where $f(x)$ is upperbounded by a constant for every constant value of $x$. Second, we gave a poly time $+2k$ approximation algorithm parameterized by vertex integrity(Vi). Complementing these positive results, we establish hardness for several parameters, including the two parameters that are smaller than the vertex cover number, vertex cover above maximum matching ($\mathrm{VC}-\mathrm{MM}$) and dominating set size, as well as the graph diameter, which is asymptotically bounded by the vertex cover number. Our results show that the problem is unlikely to be fixed-parameter tractable with respect to any of these parameters Finally, we present a $+2$-additive approximation algorithm running in time $\mathcal{O}^*(2^{O(k \log k)})$, improving upon a prior $\mathcal{O}^*(2^{O(k^2)})$-time exact algorithm for the distance-to-clique parameter, and we give a $2$-factor approximation algorithm parameterized by distance to path running in XP time.
\end{abstract}

\section{Introduction} 
Broadcasting is a fundamental communication problem that models the dissemination of information through a network. Given a graph and a source vertex initially possessing a message, the objective is to inform every vertex of the graph in as few rounds as possible, under the restriction that each informed vertex may transmit the message to at most one neighbor in every round. Despite its simple formulation, the problem captures an important trade-off between the underlying network topology and the limited communication capacity of individual vertices. Variants of the broadcasting problem have been studied in distributed computing, communication networks, and parallel algorithms, where efficient information dissemination is a central objective.

The broadcasting problem is computationally challenging on general graphs, motivating the search for efficient algorithms on restricted graph classes and under suitable parameterizations. During the past two decades, parameterized complexity has emerged as a successful framework for designing algorithms whose running time depends primarily on a structural parameter of the input rather than its overall size. Numerous graph problems that remain intractable in general admit fixed-parameter tractable algorithms when parameterized by measures such as treewidth, vertex cover number, feedback vertex set, or cluster deletion number.

\noindent \textbf{Literature.}
The Telephone Broadcast problem was introduced by Slater \textit{et al.}\cite{slater1981information} in 1981 and showed that it is NP-complete in general but can be solved in poly-time when the graph is a tree. Later, Elkin and Kortsarz\cite{elkin2002combinatorial} found that approximating it within a factor of $3-f(x)$ for any $x > 0$ is NP-hard, which rules out good polynomial-time approximations and motivates the search for exact and parameterized approaches. 

On general graphs, Fomin \textit{et al.}\cite{fomin2023parameterized} gave an exact exponential algorithm for the problem via dynamic programming over connected vertex subsets, which runs in time $3^n\cdot n^{\mathcal{O}(1)}$. This improves the earlier brute-force approach that checks all vertex permutations, obtaining a running time of $n! \cdot n^{\mathcal{O}(1)} $.

Fomin, Fraigniaud, and Golovach\cite{fomin2023parameterized} were the first to study the Telephone Broadcast problem from the perspective of parameterized complexity. They showed that the problem is fixed-parameter tractable (FPT) when parameterized by the size of a feedback edge set, the size of a vertex cover, or by $k=n-t$, where $t$ is the broadcast deadline. Their algorithm for the vertex cover parameter is based on two simple structural observations: the time needed to inform the vertex cover and the number of vertices outside the vertex cover that can be informed within a given number of rounds. They further asked whether the problem remains FPT when parameterized by treewidth, noting that no NP-hardness result was known even for treewidth at most two.

This question was negatively answered by two independent and closely related works. First, Tale\cite{tale2024double} showed that the straightforward brute-force algorithm is essentially optimal. Since the number of informed vertices can at most double in each round, the algorithm runs in $2^{2^{O(t)}} \cdot n^{\mathcal{O}(1)}$ time when parameterized by the broadcast deadline $t$. Tale\cite{tale2025telephone} further proved that, assuming the ETH, no algorithm can achieve a significantly better dependence on $t$, establishing a tight double-exponential lower bound.

Egami \textit{et al.}\cite{egami2025broadcasting} further strengthened these hardness results by proving that the problem remains NP-complete on graphs of bounded tree-depth and on cactus graphs that can be transformed into a forest of paths by deleting a single vertex. In response to these hardness results, they studied the problem under other structural parameters. They obtained FPT algorithms parameterized by vertex integrity, extending the earlier result for vertex cover, and by distance to a clique. They also designed FPT approximation algorithms parameterized by clique cover number and cluster vertex deletion.

Beyond vertex cover, related work has looked at whether even smaller structural parameters could still lead to efficient algorithms. One natural candidate is the size of a minimum dominating set, since requiring a graph to have a small dominating set is less restrictive than requiring a small vertex cover. Therefore, an FPT algorithm parameterized by the dominating set number would significantly extend the known results. However, Tale \cite{tale2025telephone} showed that this is not always possible. In particular, the problem is para-NP-hard when parameterized by several smaller structural parameters, including the feedback vertex set number and treewidth. As a result, these parameterizations do not admit XP algorithms unless ($\mathrm{P}=\mathrm{NP}$), answering an open question posed by Fomin \textit{et al.} \cite{fomin2023parameterized}. As we show in this work, the same phenomenon also holds for the dominating set number.

Polar graphs were introduced by Tyshkevich and Chernyak \cite{tyshkevich1985decomposition}as a generalization of split graphs, as well as bipartite and co-bipartite graphs, by allowing one side of the partition to induce a disjoint union of cliques rather than a single clique. Although recognizing whether a graph is polar is NP-complete, once a polar partition is given, the graph has sufficient structure to be exploited algorithmically. To the best of our knowledge, no broadcast-specific algorithm has previously been developed for polar graphs. In this paper, we present the first polynomial-time algorithm for the broadcasting problem in this graph class.

On the hardness side, recent work has shown that the boundary between tractable and intractable cases extends to very sparse, path-like graphs. Aminian \textit{et al.} \cite{aminian_et_al:LIPIcs.ICALP.2025.10} proved that the problem is NP-hard on cactus graphs of pathwidth two, while Egami \textit{et al.} \cite{egami2025broadcasting} showed that it remains NP-complete even on cactus graphs that can be transformed into a disjoint union of paths by deleting just one vertex. These results suggest that the broadcasting problem remains hard even on graphs that are very close to paths. This naturally motivates studying graph parameters that measure such closeness, such as the vertex distance to a path, which we investigate in this paper.
\subsection*{Our Contributions}
In this paper, we study the \textsc{Telephone Broadcasting} problem and advance the state of the art along several complementary axes: exact/approximate algorithms on general graphs, parameterized exact and approximate algorithms, polynomial-time solvability on a structured graph class, and hardness results. Our contributions can be summarized as follows.

\noindent\textbf{A constant additive approximation on general graphs.}
We design an algorithm that achieves a constant additive $x$ approximation to the optimum broadcast time on \emph{general} graphs, running in $\mathcal{O^*}({3-f(x)})^n)$ time, where $f(x)>0$ when $x>0$. The key algorithmic idea is a \emph{round-splitting} technique: instead of allowing a single round to inform an arbitrarily large set of uninformed vertices from the currently informed set, we show that it suffices to break a constant number of rounds into several sequential sub-rounds, each of which informs only a restricted subset of vertices. This restructuring shrinks the search space that needs to be explored, while provably costing only an additive constant number of extra rounds, yielding the improved runtime bound.

\begin{theorem}\label{thm:improved-exact}
Given an instance $(G, s, t)$ of the Broadcasting problem, there exists an exact algorithm that solves the problem in time $\mathcal{O}^*(3 - f(x))^n$ for some $x > 0$, improving upon the previously best known running time of $\mathcal{O}^*(3^n)$ due to \cite{fomin2023parameterized}.
\end{theorem}

\noindent\textbf{An improved $poly(n)$ time algorithm with parameter vertex integrity.} We give a $poly(n)$ tractable approximation algorithm with ad additive factor $+2k$, when the modulator $S$ is connected and, $k$ is the vertex integrity of the input graph. This improves upon the previously best-known FPT runtime $2^{({vi}^2log{(vi)})}.n^{O(1)}$ for this parameterization due to  ~\cite{bonnet2026fasterparameterizedbroadcasting}.

\begin{theorem}
There is a $poly(n)$ time approximation algorithm which solves \textsc{Telephone Broadcasting} problem on graphs of bounded vertex integrity with an additive $+2k$ factor, where $k$ is the vertex integrity. 
\end{theorem}

\noindent\textbf{Parameterized approximation by distance to clique.} We study the problem from a parameterized-approximation perspective and provide an algorithm that achieves an additive $2$-approximation to the optimum number of rounds when parameterized by the distance to a clique. Compared to the algorithm of Egami et al.~\cite{egami2025broadcasting}, our algorithm offers a substantially improved running time, at the cost of a mild relaxation in the approximation guarantee on the number of rounds.

\begin{theorem}
    \textsc{Telephone Broadcasting} admits a +2 additive approximation when parameterized by the distance to clique parameter $k$ with the runtime $\mathcal{O}^*(2^{\mathcal{O}(k\log k)})$.
\end{theorem}

\noindent\textbf{An XP approximation algorithm parameterized by distance to path.} We also provide a $2$-approximation XP algorithm for the problem parameterized by the distance to a path, extending efficient approximation algorithms to this structural parameter.
\begin{theorem}
    \textsc{Telephone Broadcasting} admits a $2$-factor approximation when parameterized by the distance to path parameter $k$ with the runtime $n^{\mathcal{O}(k)}$.
\end{theorem}

\noindent\textbf{Hardness results.} On the negative side, we show that the problem remains para-\textsf{NP}-hard when parameterized by (i) the vertex cover number above maximum matching size, (ii) the minimum dominating set size, and (iii) the diameter of the graph. These results delineate the boundary of tractability and indicate that it is unlikely to yield even \textsf{XP} algorithms when parameterized by these parameters, thereby complementing our positive algorithmic results.
\begin{theorem}
    \textsc{Telephone Broadcasting} is Para-NP-hard when parameterized by the vertex cover number - the cardinality of maximum matching i.e. $VC-MM$.
\end{theorem}

\begin{theorem}
    \textsc{Telephone Broadcasting} is Para-NP-hard when parameterized by the minimum dominating set.
\end{theorem}

\begin{theorem}
    \textsc{Telephone Broadcasting} is Para-NP-hard when parameterized by the diameter of the graph.
\end{theorem}

\noindent\textbf{Polynomial-time solvability on polar graphs.} We identify a graph class which we call \emph{polar graphs} in which the vertex set can be partitioned into an induced cluster subgraph and a single isolated (apex-like) vertex, with arbitrary edges permitted between the two parts. We show that \textsc{Telephone Broadcasting} is solvable in polynomial time on this class. This case is of natural practical relevance: it models settings with a single informer vertex disseminating information to entities that are themselves organized into tightly-knit communities (clusters).

\begin{theorem}
    \textsc{Telephone Broadcasting} admits a polynomial time algorithm running in time $\mathcal{O}({n^3})$ on polar graphs.
\end{theorem}

Together, these results paint a fairly complete picture of the parameterized and approximation landscape of \textsc{Telephone Broadcasting}, improving on prior work in both exact parameterized complexity and approximation guarantees, while also identifying new tractable structural classes and sharpening the known hardness frontier.


\renewcommand{\arraystretch}{1.15}
\begin{table}[H]
\centering

\hspace{2mm}
\begin{tabular}{|c|c|c|}
\hline
\textbf{Parameter}&\textbf{Results}&\textbf{Runtime}\\
 \hline
 General Graph & $+x$ Approx.  & $\mathcal{O}^*(3-f(x))^n$ \\
\hline
 Vertex Integrity  & $+2k$ Approx & $Poly(n)$ \\
 (when the modulator is connected)& & \\
\hline
Dominating Set Size &  Para-NP-hard &  \\
\hline
  VC-MM & Para-NP-hard & \\
\hline
 Diameter & Para-NP-hrad &\\
\hline
Distance to Clique  & $+2$ approximation & $\mathrm{FPT(k)}$ \\
   
\hline
 Distance to Path & $2$ factor Approx. & $\mathrm{XP}$ \\
 
\hline
Polar Graph & Exact Algorithm & $\mathrm{Poly(n)}$\\
\hline
\end{tabular}
\renewcommand{\arraystretch}{1}
\caption{Overview of our results}
\end{table}


\section{Preliminaries and Notations}

\noindent\textbf{\textsc{Graphs.}} We follow standard graph-theoretic notation and terminology, referring to the textbook of Diestel\cite{diestel2024extremal} for any undefined notions. Unless stated otherwise, all graphs considered are finite and undirected. We write $V(G)$ and $E(G)$ for the vertex and edge sets of a graph $G$.  For a set of vertices $X$, the \emph{induced subgraph} $G[X]$ keeps exactly the vertices of $X$ and the edges of $G$ with both endpoints inside $X$; we abbreviate $G - X = G[V(G) \setminus X]$. The \emph{neighborhood} of a vertex $v$ is $N(v) = \{u : \{u,v\} \in E(G)\}$, and its degree is $|N(v)|$. 
 
A \emph{clique} is a set of pairwise adjacent vertices, and an \emph{independent set} is a set of pairwise non-adjacent ones. A \emph{matching} is a set of edges no two of which share an endpoint, and we say a matching \emph{saturates} a set $A$ if every vertex of $A$ is the endpoint of one of its edges. Matchings are the natural language for a single round of broadcasting informing a fresh batch of vertices in one step is precisely matching each new vertex to a distinct informer and this is the connection Observation~\ref{obs:one-round-matching} makes explicit. Finally, the \emph{diameter} of $G$ is the largest shortest-path distance between any two of its vertices.
A \emph{vertex cover} is a set of vertices that contains at least one endpoint of every edge. The size of a minimum vertex cover is denoted by $\text{vc}(G)$. Since each edge in a matching must be covered by a different vertex, the maximum matching number $\text{mm}(G)$ satisfies $\text{mm}\le \text{vc}$. We study the parameter VC-MM(vertex cover over maximum matching) which measures the gap between these two quantities. A \emph{dominating set} is a set $D$ such that every vertex outside $D$ has a neighbor in $D$. Its minimum size is denoted by $\text{dom}(G)$. Since every vertex cover is also a dominating set in graphs without isolated vertices, $\text{dom}(G)\le \text{vc}(G)$.

We also consider graph parameters defined by a \emph{modulator}, that is, a set of vertices whose removal leaves the graph in a simpler graph class. If $G-M$ is a path, then $|M|$ is called the \emph{distance to a path}. Similarly, if $G-M$ is a clique, then $|M|$ is called the \emph{distance to a clique}. In both cases, the parameter $k=|M|$ measures how far the graph is from the corresponding graph class.

\noindent\textbf{Twins and polar graphs.} Two vertices in an independent set are called \emph{twins} if they have the same neighborhood.  Therefore, if $S$ is a vertex cover, the independent set can be partitioned into at most $2^{|S|}$ neighborhood types. A graph is \emph{polar} if its vertex set can be partitioned into two parts, one inducing a disjoint union of cliques and the other an independent set. Such a partition is called a \emph{polar partition}. Polar graphs generalize several well-known graph classes, including split, bipartite, and co-bipartite graphs\cite{tyshkevich1985decomposition}. Since recognizing polar graphs is NP-complete, we assume that a polar partition is given as part of the input in Section~\ref{sec:polar}.
Throughout, $\log$ means the base-$2$ logarithm, and $[p]$ stands for $\{1, 2, \dots, p\}$ for $p\in \mathbb{N}$.

\medskip

\noindent\textbf{The broadcasting model.} Let $G$ be a connected graph and let $s \in V(G)$ be a distinguished source vertex that initially holds a single piece of information (equivalently, a message). Information is spread in a sequence of rounds. In each round, an informed vertex may inform at most one uninformed neighbor, and an uninformed vertex can receive the information from at most one informed neighbor. It also can be represented by a sequence of informed sets $S_0 \subseteq S_1 \subseteq \cdots \subseteq S_\tau$, where $S_i$ is the set of vertices informed by the end of round $i$. Initially, $S_0=\{s\}$ and finally $S_\tau=V(G)$. For every $i$, the newly informed vertices $S_{i+1}\setminus S_i$ must be matchable to distinct vertices in $S_i$ by a matching in the bipartite graph induced by $(S_i,\, S_{i+1}\setminus S_i)$. Similarly, each vertex in $S_{i+1}\setminus S_i$ is informed by a distinct vertex in $S_i$, reflecting the telephone broadcasting rule.
 
The second picture is a \emph{broadcast tree}: a spanning tree rooted at $s$ in which each vertex, once informed, hands the message to its children one per round in some chosen order. We formally define, a \emph{broadcasting protocol} is a pair $\bigl(T, \{C(v) \mid v \in V(T)\}\bigr)$, where $T$ is a spanning tree of $G$ rooted at $s$, and for each $v \in V(T)$, $C(v)$ is an ordered set of the children of $v$ in $T$. As soon as $v$ receives the message, it starts sending it to its children in $T$ in the order defined by $C(v)$.
 
The \emph{length} of a broadcasting protocol is the number of rounds it uses. The minimum number of rounds required to inform all vertices from a source vertex $s$ is called the \emph{broadcast time} of $G$ from $s$, denoted by $b(G,s)$. A protocol that achieves this minimum is called an \emph{optimal} broadcasting protocol. We write $T_{opt}$ for $b(G,s)$ when the source is clear from the context. If the source vertex is not fixed, we consider the worst-case broadcast time, defined as $b(G)=\max_{u\in V(G)} b(G,u)$. Since each informed vertex can inform at most one new vertex in a round, the number of informed vertices can at most double in every round. Therefore, $b(G,s)\ge \left\lceil\log_2 n\right\rceil.$ This bound is tight, as it is achieved by complete graphs.
 
The decision problem we study packages this up as follows.
 
\noindent\textbf{\textsc{Telephone Broadcast}.} Given a connected graph $G$, a source $s$, and an integer $t$, decide whether the message can reach every vertex within $t$ rounds, that is, whether $b(G,s)\le t$.

We call such an instance a \emph{YES-instance} if the deadline can be met, and a \emph{NO-instance} otherwise. There is also a \emph{source-free} version that asks whether $b(G,s) \le t$. Both are NP-complete on general graphs~\cite{slater1981information}, although they become easy on trees; worse, even approximating the broadcast time within a factor $3 - f(x)$ is NP-hard for every $x > 0$~\cite{elkin2002combinatorial}, which is a large part of why parameterized and structural approaches are so appealing here.
 
\noindent\textbf{Parameterized complexity.} We use parameterized complexity in the usual way (see~\cite{fomin2023parameterized} for its use on broadcasting). An instance carries both a size and a \emph{parameter} $k$, and the question is how the difficulty scales with $k$ rather than with the whole input. A problem is \emph{fixed-parameter tractable} (FPT) if it can be solved in time $f(k) \cdot n^{O(1)}$ for some computable function $f$ depending only on $k$ \cite{Cygan2015Parameterized}. We write $O^*(\cdot)$ when hiding polynomial factors keeps the expression readable. A problem belongs to the class $\mathrm{XP}$ if it can be solved in time $n^{f(k)}$, where $k$ is the parameter. Thus, for every fixed value of $k$, the algorithm runs in polynomial time. On the other hand, a problem is \emph{para-NP-hard} if it remains NP-hard even when the parameter is fixed to a constant. Consequently, such a problem can't admit an $\textsf{XP}$ algorithm unless $\mathrm{P}=\mathrm{NP}$.
 
\noindent\textbf{Parameters.} Throughout this paper, we consider several well-studied structural graph parameters that capture different aspects of a graph's complexity. These parameters play a central role in our algorithmic and hardness results and are related by well-known inequalities. We briefly define the parameters used in this work and summarize the relationships between them.

\section{Approximation Algorithm for the Broadcasting Problem on General Graph}
In this section, we present an exact algorithm for the Broadcasting problem that runs in time $(3-f(x))^n$ for some $x > 0$. This improves upon the previously best known algorithm due to \cite{fomin2023parameterized}, which solves the Broadcasting problem exactly in time $3^n$. Before we describe the algorithm, we discuss the overall idea and make the following observation which will be used crucially in the construction.

\subsection{A Brief Overview}
For a given instance of the Broadcasting problem $(G, s, t)$, we construct a new directed graph $G' = (V', E')$ such that each vertex $v_S \in V'$ corresponds to a subset $S \subseteq V(G)$. Recall that in the Broadcasting problem, we wish to broadcast information from the source $s$ to all vertices of $G$ within time $t$. This is equivalent to finding an optimal sequence $\{S_1, S_2, \ldots, S_t\}$ of subsets of $V(G)$, where $S_1 = \{s\}$, $S_t = V(G)$, and $S_i \subsetneq S_j$ for all $i < j$, such that every vertex $v \in S_i$ has received the information by round $i$.

The key intuition behind the construction of $G'$ is to think of each vertex $v_S \in V'$ as a \emph{state}, representing the set of vertices that have been informed so far. We connect these states in such a way that reaching state $v_{S_i}$ at step $i$ implies that all vertices in $S_i$ have received the information within round $i$. Therefore, finding an optimal broadcasting sequence for $G$ reduces to finding a shortest path in $G'$ from the initial state $v_{\{s\}}$ to the final state $v_{V(G)}$.

As we will show, this approach yields an $+x$ approximation algorithm, running in time strictly less than $(3-f(x))^n$ (for some positive increasing function $f(x)$ w.r.t. the constant $x$). This gives us our main theorem, which we prove in the remainder of this section.

\subsection{The Algorithm}

Before we explain our algorithm in two parts, we make the following observation~\ref{obs:one-round-matching}.

\begin{observation}\label{obs:one-round-matching}
For any two subsets $S, S' \subseteq V(G)$ with $|S| < |S'|$, it is possible to broadcast from $S$ to $S'$ in some round $i$ if and only if there exists a matching in the bipartite graph $(S, S' \setminus S)$ saturating all vertices $v \in S' \setminus S$, where the edge set is induced by $G$.
\end{observation}

\noindent \textbf{1. Construction of $G'$:} For a given instance $G$ with $n$ vertices, we consider all possible subsets of $V(G)$. For each subset $S \subseteq V(G)$, we introduce a vertex $v_S$ in $G'$. Thus $G'$ has exactly $2^n$ vertices in total, and we think of each vertex $v_S$ as a \emph{state} representing the set $S$ of informed vertices. For each ordered pair of subsets $(S, S')$, we add a directed arc from $v_S$ to $v_{S'}$ in $G'$ if and only if the following three conditions are satisfied.
\begin{itemize}
    \item[\textbf{(C1)}] $S \subset S'$.
    \item[\textbf{(C2)}] Either $|S| \leq n/x$, or 
    $|S' \setminus S| \leq n/x$. 
    \item[\textbf{(C3)}] There exists a matching saturating $S' \setminus S$ in the bipartite graph with bipartition $(S, S' \setminus S)$ whose edge set is induced by $G$.
\end{itemize}
We briefly justify each condition. In the Broadcasting problem, at every intermediate round $t' \leq t$, there is always at least one new vertex that receives the information in round $t'$ but not in round $t' - 1$. Condition (C1) ensures that the set of informed vertices is strictly increasing across any two consecutive rounds. Condition (C2) is imposed to bound the total number of arcs in $G'$, which in turn controls the running time of the algorithm, as we will analyze later. By Observation \ref{obs:one-round-matching}, condition (C3) ensures that adding an arc from $v_S$ to $v_{S'}$ is valid, in the sense that it is actually possible to broadcast from the state $v_S$ (where all vertices in $S$ are informed) to the state $v_{S'}$ (where all vertices in $S'$ are informed) in exactly one round.

\noindent \textbf{2. Finding the Shortest Path in $G'$:} In the second part of the algorithm, we find a shortest path in $G'$ from the initial state $v_{\{s\}}$ to the final state $v_{V(G)}$ using the Breadth First Search (BFS) algorithm. Since BFS runs in time proportional to the number of vertices and edges of the graph it is applied to, this step runs in time $|V(G')|+|E(G')|$.

\subsection{Proof of Theorem 1.1}

Towards proving the main result of this section, we establish the following two lemmas.

\begin{lemma}\label{lem:degree-bound}
The maximum degree of the graph $G'$ is bounded by ${O}^*\left(\left(\frac{1}{\left({\frac{1}{x}}\right)^\frac{1}{x}\left({1-\frac{1}{x}}\right)^{\left(1-\frac{1}{x}\right)}}\right)^n\right)$.
\end{lemma}

\begin{proof}
For any pair $(S, S')$, conditions (C1) and (C3) together require that $S \subset S'$ and that there exists a matching saturating $S' \setminus S$ in the bipartite graph $(S, S' \setminus S)$ induced by $G$. Since a matching $M$ saturating $S' \setminus S$ must be contained in $S$, so we need $|S' \setminus S| \leq |S|$.
Note that, for any set $v_S$, it has an arc (to and from) with any $v_{S'}$ that satisfies $|S|\leq n/x$ or $|S'\setminus S|\leq n/x$. Therefore, the number of arcs associated with $S$ is essentially the number of such sets ($S'$). Further, such sets $S'$ different from $S$ at most $n/x$ entries (either they have excess or deficit). The number of such excess or deficit sets is bounded by $\sum_{i=1}^{i=n/x}\binom{n}{i}$.
\begin{claim}
     $\sum_{i=1}^{i=n/x}\binom{n}{i}$ is bounded by ${O}^*\left(\left(\frac{1}{\left({\frac{1}{x}}\right)^\frac{1}{x}\left({1-\frac{1}{x}}\right)^{\left(1-\frac{1}{x}\right)}}\right)^n\right)$, whenever $x\geq 2$.
\end{claim}

\begin{proof}
    We use the following bound from Stirling's approximation \cite{robbins1955remark}
    $$\sqrt{2\pi}(n)^{n+1/2}e^{-n}e^{\frac{1}{12n+1}}\leq n!\leq\sqrt{2\pi}n^{n+1/2}e^{-n}$$ 

    We prove the bound for the largest term in the summation. Once we do this, the lemma holds (by multiplying the upper bound for the largest term in the summation by $n$).

    Previously stated bounds mean $\binom{n}{n/x}\leq \frac{1}{\sqrt{2\pi}(1/x)^{(n/x)}(1-1/x)^{n(1-1/x)}}\times \frac{1}{\sqrt{n}e^{1/6n}\sqrt{1/x(1-1/x)}}$
    
    However, as $n$ increases , it is easy to check that the later factor converges to 1. Therefore, we have the required bound.
\end{proof}

\end{proof}

\begin{lemma}\label{lem:broadcast-to-path}
    If there is a broadcasting sequence for $G$ of size $k$, then there exits a shortest path in $G'$ of length at most $k+x$.
\end{lemma}
\begin{proof}
    To prove this, we construct a path in the graph $G'$ utilizing the optimal broadcasting sequence $\{s\}=S_0\subset S_1 \subset S_2 \subset\cdots \subset S_k=\{V(G)\}$. Since the sequence is strictly increasing, the sets $S_0, S_1, \ldots, S_k$ are all distinct and therefore the corresponding vertices $v_{S_0}, v_{S_1}, \ldots, v_{S_k}$ are all distinct in $G'$. Conditions (C1) and (C3) hold at every consecutive pair $(v_{S_i}, v_{S_{i+1}})$ by the definition of a broadcasting sequence and Observation~\ref{obs:one-round-matching} respectively. If condition (C2) also holds at every step, then the sequence $v_{S_0}, v_{S_1}, \ldots, v_{S_k}$ is directly a valid path of length $k$ of distinct vertices in $G'$, and we are done.
    
    Otherwise, suppose (C2) fails at some step $i$, meaning there is no direct arc between $v_{S_i}$ and $v_{S_{i+1}}$ in $G'$. In this case, we route through intermediate dummy states.
    
    We break this sequence as follows.
    \begin{enumerate}
        \item If $|S_{i+1}\setminus S_{i}|>n/x$ and $|S_i|> n/x$, we partition $S_{i+1}\setminus S_{i}$ into $\lceil|S_{i+1}\setminus S_{i}|\cdot(x/n)\rceil$ many sets, each with size at max $n/x$, call them $S_{i1},S_{i2},\cdots,S_{i\lceil|S_{i+1}\setminus S_{i}|\cdot (x/n)\rceil}$. We add these sets sequentially between $S_i$ and $S_{i+1}$. 
        \item If $|S_{i+1}\setminus S_{i}|\leq n/x$ or $|S_i|\leq n/x$, we add all possible edges.
    \end{enumerate}

    Once we have this sequence, we rename the sets starting from $S_0$ till the last set $S_{k}$ along with the intermediate newly added states (call $Q_i$) in between. So $S_0=Q_0\subset Q_1\cdots \subset Q_{k'}=S_k$
    
    Note that the new sequence has at most $x$ many additional sets. Since, the number of blocks is $$ \left\lceil \frac{|S_{i+1}\setminus S_i|}{n/x} \right\rceil = \left\lceil |S_{i+1}\setminus S_i| \cdot \frac{x}{n} \right\rceil. $$
    
Since $S_{i+1}\setminus S_i$ is a set of vertices, $|S_{i+1}\setminus S_i| \le n$. Plugging in the largest possible value $|S_{i+1}\setminus S_i| = n$:
$$
\left\lceil n \cdot \frac{x}{n} \right\rceil = \lceil x \rceil = x.
$$

So the number of blocks is maximized when the newly informed set is as large as possible (all $n$ vertices), and even then it is exactly $x$. For any smaller $|S_{i+1}\setminus S_i|$, the count $\left\lceil |S_{i+1}\setminus S_i| \cdot (x/n) \right\rceil$ is smaller. Hence, the number of blocks and, therefore the number of new states added is at most $x$. Since each set later in the sequence is a superset of the one before and each set $Q_i$ has a matching with $Q_{i+1}$ saturating $Q_{i+1}\setminus{Q_i}$ of size at most $n/x$, we have that there is an edge between the vertices $v_{Q_i}$ and $v_{Q_{i+1}}$. This ensures that we have a path $\{v_{Q_0},v_{Q_1},\cdots v_{Q_k'}\}$ of edge length $k'\leq k+x$. This completes the proof.
\end{proof}

\begin{theorem}
Broadcasting in general graphs admits an additive $+x$ approximation in time $O^{*}\big((3 - f(x))^{n}\big)$ for some $x > 0$.
\end{theorem}

\section{A $\mathrm{OPT}+2k$ Approximation for Broadcast under Vertex Integrity}

\paragraph{Vertex integrity modulator.}
A set $S\subseteq V(G)$ is a \emph{modulator of order $k$} if
$|S|+\max_{C}|C|\le k$, where $C$ ranges over the connected components of $G-S$.
Write $a=|S|$ and let $b=\max_{C}|C|$, so $a+b\le k$ (hence $a\le k$ and
$b\le k$). Let $C_1,\dots,C_m$ be the components of $G-S$, and let $C_0$ denote
the one containing $s$ (if $s\in S$, there is no $C_0$).

\paragraph{Standing assumptions.}
$G$ is connected and $S$ is connected.

\subsection*{The Algorithm}

The schedule runs in three phases.

\noindent \textbf{Phase 1}(Inform $S$). Route the message from $s$ to $S$ and flood $S$.

\noindent \textbf{Phase 2} (Seed the components). Compute the smallest integer $t$ for which there is an assignment $\phi$ sending each non-source component $C_i$ to an adjacent vertex of $S$, with no vertex of $S$ receiving more than  $t$ components. (Binary-search on $t$; feasibility for a fixed $t$ is a bipartite $b$-matching, solvable by max-flow.) Then, over $t$ rounds, each vertex $v\in S$ calls the components assigned to it, one per round, planting one informed vertex in each.

\noindent \textbf{Phase 3}(Fill the components). Inside each component, flood from its seed.

All three phases run in polynomial time.

\begin{lemma}[Phase 1 costs $\le k-1$]
Phase 1 informs all of $S$ within $k-1$ rounds.
\end{lemma}

\begin{proof}
If $s\in S$, flooding the connected set $S$ takes $\le a-1$ rounds. Otherwise
$s\in C_0$; since $G$ is connected, $C_0$ has an edge to $S$. The message travels
inside $C_0$ to a boundary vertex ($\le |C_0|-1$ hops), crosses into $S$
($1$ hop), then floods $S$ ($\le a-1$ rounds). In total
$(|C_0|-1)+1+(a-1)=|C_0|+a-1\le b+a-1\le k-1$.
\end{proof}

\begin{lemma}[The assignment value bounds $\mathrm{OPT}$]
The minimum feasible $t$ in Phase 2 satisfies $t\le\mathrm{OPT}$.
\end{lemma}

\begin{proof}
Fix any optimal schedule. Take a component $C_i$ not containing $s$. Its first
informed vertex must have been called by a neighbor outside $C_i$; every such
neighbor lies in $S$, because the components of $G-S$ are pairwise
non-adjacent. So each non-source component has an adjacent vertex of $S$ that
first informed it. A vertex of $S$ makes one call per round, so across
$\mathrm{OPT}$ rounds it is the first informer of at most $\mathrm{OPT}$
components. Reading this ``first-informer'' map off the optimal schedule gives an
adjacency-respecting assignment with every $S$-vertex used at most
$\mathrm{OPT}$ times --- a feasible assignment for $t=\mathrm{OPT}$. Hence the
minimum feasible $t$ is $\le\mathrm{OPT}$.
\end{proof}

\begin{lemma}[Phase 3 costs $\le k-1$]
Once each component holds one informed vertex, flooding all components takes
$\le k-1$ rounds.
\end{lemma}

\begin{proof}
Each component is connected with at most $b\le k$ vertices, so it floods from its
seed in $\le b-1\le k-1$ rounds. Distinct components share no edges, so they
flood simultaneously.
\end{proof}

\begin{theorem}
The algorithm outputs, in polynomial time, a broadcast schedule of length at
most $\mathrm{OPT}+2k$.
\end{theorem}

\begin{proof}
By Lemmas 1 and 3 the outer phases cost at most $k-1$ each, and Phase 2 costs
$t$. The total is $(k-1)+t+(k-1)=t+2k-2$. By Lemma 2, $t\le\mathrm{OPT}$, so the
length is at most $\mathrm{OPT}+2k-2\le\mathrm{OPT}+2k$. Correctness of the
schedule is immediate: after Phase 1 all of $S$ is informed; after Phase 2 every
component holds a seed; after Phase 3 every vertex is informed.
\end{proof}

\noindent\textbf{Remark.} Phase 2 seeds only the non-source components, since $C_0$ already contains an informed vertex and is filled in Phase 3.

\begin{theorem}
There is a $poly(n)$ time approximation algorithm which solves \textsc{Telephone Broadcasting} problem on graphs of bounded vertex integrity with an additive $+2k$ factor, where $k$ is the vertex integrity. 
\end{theorem}

As shown earlier in the paper \cite{fomin2023parameterized}, the Broadcasting problem is fixed-parameter tractable (FPT) when parameterized by the minimum vertex cover size. This naturally raises the question of whether an FPT algorithm exists for parameters slightly smaller than the vertex cover size. Here, we answer this negatively by presenting two para-NP-hardness proofs, each obtained via a reduction from a general Broadcasting graph instance to a Broadcasting instance with a constant value of, respectively, vertex cover above maximum matching and dominating set size, as well as diameter of the graph.

\subsection{\textsf{Para-NP}-hardness w.r.t. VC over MM}

Since the problem is \textsf{FPT} when parameterized by the vertex cover number of the graph, we consider a smaller related parameter, namely the vertex cover number above the maximum matching ($VC-MM$). For this choice of parameter, we show that the problem is \textsf{XP-}intractable, which is to say, for a fixed value of the parameter, finding a polynomial time algorithm is impractical.  

 We do this by a polynomial reduction of an arbitrary (NP-hard) instance of the \textsc{Telephone Broadcasting} to another related instance of the \textsc{Telephone Broadcasting} where $VC-MM=1$. Any \textsf{XP} algorithms for the later therefore imply a polynomial algorithm for the former (which would violate the NP-hardness). We now start with the reduction.

\noindent\textbf{Construction.} Consider a decision version of the \textsc{Telephone Broadcasting} problem-$(G,s,t)$ where $G$ is the input graph, $s$ is the unique vertex that has the information at the beginning and $t$ is the time (number of rounds) required to inform all the vertices in the graph. $(G,s,t)$ is a ``YES" instance of the problem if it is possible to disseminate the piece of information from $s$ to all the vertices in $G$ within $t$ rounds and a ``NO"-instance otherwise.

From $(G,s,t)$, we construct another instance $(G',s,t')$ of the \textsc{Telephone Broadcasting} problem as follows
\begin{enumerate}\label{newinstance}
     \item $V(G')=V(G)\cup V'$ where $V'$ is a copy of $V(G)\setminus \{s\}$. We call the copy of a vertex $v\in V(G)$ as $v'\in V'\subset {V(G')}$.
     \item $E(G')=E(G)\cup \{(v,v'):v\in V(G)\setminus \{s\}\}$, i.e, each vertex in $V(G)$ except the source $s$ is adjacent to a unique copy of itself in $V'$.
     \item $t'=t+1$
 \end{enumerate}

\begin{observation}\label{vc-mmboundedparameter}
    The parameter $VC-MM$ for the graph $G'$ constructed in \ref{newinstance} from a graph $G$ is 0 (i.e. bounded by a constant).
\end{observation}
\begin{proof}
    Consider the set of edges - $M=\{(v,v'):v\in V(G)\setminus \{s\}\}$. Here, note that the edges are mutually non adjacent, i.e., they form a matching. The only vertex not matched by any of the edges in $M$ is $s$. Therefore, the matching $M$ is a maximum matching. 

    Again, since we know that the minimum vertex cover number is at least the cardinality of the maximum matching\cite{diestel2012graph}, we have $VC\geq MM$. Moreover, note that any edge of the graph $G'$ is incident to at least a vertex in $V(G)\setminus \{s\}$ (i.e., of the same cardinality as the maximum matching). These two properties guarantee that the set $V(G)\setminus \{s\}$ is indeed a minimum vertex cover. Therefore, for the constructed graph $G'$, $VC-MM=0$.
\end{proof}

Now, the only property left to check is that the reduction \ref{newinstance} is indeed a valid polynomial reduction. To this end, we mention the following lemma
\begin{lemma}
     The reduction \ref{newinstance} is a valid polynomial reduction, i.e., $(G',s,t')$ can be constructed in polynomial time from $(G,s,t)$ and a YES instance of the former implies a YES instance of the later and vice versa.
\end{lemma}

\begin{proof}
     Since $G'$ is constructed from $G$ by copying its edges and vertices and then adding private neighbors to all but one vertices of $G$, the construction is done in polynomial time. We therefore focus on the correctness of the reduction.

     We first show that if $(G,s,t)$ is a YES instance, then the constructed instance $(G',s,t')$ is also a YES instance. To this end, we give a broadcasting strategy for $(G',s,t')$ from $(G,s,t)$ instance. For the first $t$ time step, we broadcast the information as in $G$ (i.e. as in the instance $(G,s,t)$). However at the last round (the +1 round), we inform all the other vertices, namely the private neighbors of the vertices $V(G)$ in the graph $G'$ (this is valid as $M=\{(v,v'):v\in V(G)\setminus \{s\}\}$ forms a matching as argued earlier).

     Finally, we show that if $(G',s,t')$ is a YES instance then so is $(G,s,t)$. To this end, note that in the graph $G'$, any path from the vertex $s$ to the vertex $v'\in V'$ has to pass through $v\in V(G)$, because $v'$ is only adjacent to $v$. Therefore, by the time $t'-1$, all the vertices in $V(G')\setminus V'=V(G)$ has to be informed, otherwise assume, $v\in V(G')\setminus V'$ is not informed by the time $t'-1$, since $v'$ is a private neighbor of $v$, $v'$ cannot be informed by the time $t'$ (as the only neighbor of $v'$ is $v$, which itself is assumed not to be informed by the time $t-1$). Therefore, we have that by the time $t'-1=t$ all the vertices in $V(G)\setminus V'$ are informed, which is the same as saying $(G,s,t)$ is a YES instance. 
\end{proof}
Therefore, the reduction is indeed a valid polynomial reduction, and we have the following hardness result.
\begin{theorem}
    \textsc{Telephone Broadcasting} is Para-NP-hard when parameterized by the vertex cover number - the cardinality of the maximum matching i.e., $VC-MM$.
\end{theorem}

\subsection{\textsf{Para-NP}-hardness w.r.t. Dominating set and Diameter}
Our second choice of parameter is dominating set size. Again with the same idea we show that the problem is W-hard to solve by showing a polytime reduction from an arbitrary (or general) Broadcasting instance, as it is maintained earlier that, it is NP-hard to decide a general broadcast instance $G(V, E,s,t)$ is a yes instance or not. Here we construct a new broadcasting instance with dominating set size $1$.

\noindent \textbf{Construction:}
Given an instance $G(V,E,s,t)$ of the Broadcasting problem, we construct a new instance $G'(V',E',s,s_1,t+1)$ as follows:
$$V' = V \cup \{s_1\} \cup \{s_{11}, s_{12}, \dots, s_{1t}\}$$
$$E' = E \cup \{(s_1,v) \mid \forall v \in V'\setminus\{s_1\}\}$$
That is, we introduce $t+1$ many new vertices: a vertex $s_1$ and $t$ degree one many vertices $s_{11}, s_{12}, \dots, s_{1t}$. The vertex $s_1$ is made adjacent to every other vertex in $V'$ (i.e., $s_1$ is a universal/globally connected vertex), while each $s_{1i}$, for $1 \le i \le t$, is connected only to $s_1$. Since $|V'| = |V| + t + 1$ and $|E'| = |E| + |V| + t$, the construction can clearly be made in polynomial time. Also, notice $s_1$ is adjacent to every other vertex of $G'$, the singleton set $\{s_1\}$ is a dominating set of $G'$. Hence, the domination number of $G'$ is $1$.

\begin{lemma}\label{lem:yes-instance-equivalence}
The instance $G(V, E,s,t)$ is a YES instance of the Broadcasting problem if and only if the instance $G'(V', E',s,s_1,t+1)$ is a YES instance of the Broadcasting problem.
\end{lemma}
\begin{proof}
($\Rightarrow$) Suppose $G(V,E,s,t)$ is a yes-instance.

Then there exists an optimal broadcasting protocol $\left(T, \{C(v) \mid v \in V(T)\}\right)$ that informs all vertices of $G$ within $t$ time steps. We construct a broadcasting protocol for $G'$ as follows.

In the first time step, $s$ informs $s_1$ (instead of following its original first-round action). From the second time step onward, we simulate the original protocol on $G$ exactly as before to inform all vertices of $V$. Simultaneously, since $s_1$ is now informed after round $1$, it uses its remaining $t$ time steps to inform its $t$ neighbors $s_{11}, s_{12}, \dots, s_{1t}$, activating exactly one such neighbor per time step.

Since the vertices of $V$ are informed within the original $t$ rounds (now occurring in rounds $2$ through $t+1$), and the $t$ vertices $s_{1i}$ are informed one per round over the same span, all vertices of $G'$ are informed within $t+1$ time steps. Hence, $G'(V', E',s,s_1,t+1)$ is a yes-instance.

\noindent($\Leftarrow$) Suppose $G'(V',E',s,s_1,t+1)$ is a yes-instance.

We claim that in any valid broadcasting protocol for $G'$ completing within $t+1$ steps, the vertex $s_1$ must be informed at time step $1$.

Indeed, each vertex $s_{1i}$ is adjacent only to $s_1$, so $s_1$ is the only vertex capable of informing $s_{1i}$. Since there are $t$ such vertices $s_{11}, \dots, s_{1t}$, and each informed vertex can inform at most one new vertex per time step, informing all $t$ of them requires $t$ full time steps after $s_1$ itself becomes informed. As only $t+1$ time steps are available in total, $s_1$ must already be informed by the end of time step $1$- otherwise the $s_{1i}$'s cannot all be informed within the allotted time.

Since $s_1$ becomes informed at time step $1$ and is only adjacent to $s$ and vertices of $V'$ (via its universal connections), and since $s_1$ must devote all of its remaining $t$ time steps to informing the $s_{1i}$'s, $s_1$ cannot participate in informing any vertex of $V$. Consequently, all vertices of $V$ must be informed entirely through the original source $s$, using the remaining $t$ time steps (time steps $2$ through $t+1$).

This is exactly a valid broadcasting protocol for $G$ completing in $t$ time steps. Hence, $G(V, E,s,t)$ is a yes-instance.
\end{proof}

       \begin{theorem}
           \textsc{Telephone Broadcasting} is Para-NP-hard when parameterized by the dominating set size.
       \end{theorem}

       Further, note that the diameter of the constructed graph is bounded by $2$, therefore, we have-

       \begin{theorem}
           \textsc{Telephone Broadcasting} is Para-NP-hard when parameterized by the diameter of the graph.
       \end{theorem}

\section{An \textsf{XP} Algorithm Parameterized by (vertex) Distance to Path}

Here in this section we have a modulator set of $M$ of atmost $k$ many vertices and a path $P$ of $n$ many vertices. We build an algorithm sequentially for this special graph structure. Initially, for the exposition, we provide a clear algorithm for the case when $k$ (size of the path modulator) is 1 and then use this algorithm (essentially a $\DP$ procedure to solve the general case in \textsf{XP} time). Note that, we provide these algorithms for the case when the information initially resides with a modulator vertex, finally in the end we argue why the other case is not drastically different from this.

Throughout this section, we use $P=\{p_1,p_2,\cdots, p_{|P|}\}$ to denote the path that results when we delete all the modulator vertices and $M$ to denote all the modulator vertices. Hence $E(G)=\{(u,v)|u,v \in M\} \cup \{(u,v)|u,v \in P\} \cup \{(u,v)|u \in M, v \in P\}$ and $V(G)=V(M)\cup V(P)$.

\subsection{Case 1($k=1$)}
Given $M=\{v\}$ and a path $P$ of $n$ vertices. We need to find an optimal broadcasting protocol for $G[M \cup P]$.
Since we have a path and an additional vertex $v$ (that is adjacent to some of the vertices in the path), we may assume that whenever a vertex of the path is informed from $v$ (assuming $v$ as the source vertex), it transmits the information in the next step to an adjacent uninformed vertex on the path .

\noindent \textbf{Covering Analogy:} Note that informing a vertex in the path at a step $i$, will help to transmit the information to the $t-i$ vertices on one side and the $t-i-1$ vertices on the other. Therefore, the problem of asking if the entire path can be informed in $t$ many steps is same as asking if the entire path can be covered by paths of length $1,2,4,6,8,\cdots 2t-2$ by placing their central vertex at one of the neighbors of $v$.\\
Roughly this covering analogy is what we follow for the rest of the results in this section.


\begin{observation}\label{obs:path-replacement}
In the covering analogy, if a collection of subpaths, each of length $i \leq 2t-2$ and centered at a vertex adjacent to $v$, covers the entire path $P$, then replacing each such subpath of length $i$ by a subpath of length at least $2t-1$ (centered at the same vertex) also covers the entire path $P$.

\end{observation}

To give a 2-factor approximation algorithm, first note-

\begin{observation}\label{2factobser}
    Informing a path vertex in the first $t$ (out of the available $2t=2T_{\text{opt}}$) time would result in informing at least $t$ vertices on either side of the vertex. Therefore, we have $t$ many paths of length at most $2t-1$, each of whose central vertex can be placed on one of the neighbors of $v$. Therefore, having twice the amount of optimum time ($2t$) guarantees us to have the same number of sub paths (as in the case where the time available was $t$). Moreover the length of each of these paths can be assumed to be at least $2t-1$
\end{observation}
Therefore, we have, from the previous observations, that the entire path can be covered by informing some vertices in the first $t$ steps out of the total $2t$ steps. To find the exact vertices who are informed in the first $t$ steps, we use a simple greedy procedure. \\
    
We proceed by placing a path of length $2t-1$ centered at the rightmost neighbor of $v$ that allows for all the uninformed vertices to the left of $v$ to lie within $t-1$ vertex distance. We do this iteratively by placing at most $t$ many paths of length $2t-1$. This completes our greedy algorithm.


We now argue that this greedy procedure is correct and that it achieves the claimed factor. Let $T_{\text{opt}}$ be the optimal broadcast time, i.e. the fewest steps in which $v$ can inform the whole path. By Observations~\ref{obs:path-replacement} and \ref{2factobser}, we may assume every path we use in the covering has length $2T_{\text{opt}}-1$. In other words, once a neighbor of $v$ is informed, it can reach up to $t^\ast$ vertices on each side.
    
The key point is that greedy choice is always safe.
\begin{lemma}[Exchange]\label{lem:exchange}
  Considering the optimum placement of the sub path (where the total available time was $T_{\text{opt}}$), our greedy procedure covers a super set of vertices of the path when considering the same number of rightmost placed sub paths.
\end{lemma}
\begin{proof}
Consider $k$ to be the number of rightmost placed paths considered.

We prove that the lemma is true for each $k\leq T_{\text{opt}}$. To this end, we use induction.

    $k=1$ is easy to prove since the length of each of the paths is longer in our ($2T_{\text{opt}}$ setting).

    We assume that the lemma statement is true when $k=l$, we prove the lemma when $k=l+1$. To do this, since we know that the length of the path covered by the first $l$ leftmost sub paths in our case is a super set of what is covered in the optimal case, we could just place the center of $l+1^{\text{th}}$ sub path of length $2T_{\text{opt}}-1$ at the same vertex where it is placed in the optimal case. Since the length of the sub path in the optimal case is strictly smaller than $2T_{\text{opt}}-1$, we are covering a super set (say $S_2$) of what is covered in the optimal case (say $S_1$ ). Further note that in the greedy step, we are choosing the rightmost vertex to place the center of $l+1^{\text{th}}$ path, leaving no uninformed vertex on the left side, this guarantees us that we are covering (in our greedy procedure) a super set $S_3$ of $S_2$. Hence, the lemma statement is true by induction.
    \end{proof}

\begin{theorem}\label{thm:2-factor-approx}
    There is a 2-factor approximation for broadcasting the information from $v$ to all the vertices in the graph.
\end{theorem}

    We now sketch how this idea can be extended when the path modulator size is $k$. For this, we provide an \textsf{XP}-algorithm. Essentially in the extra ``\textsf{XP}" time, we guess ``how" the modulator vertices are informed. Once this is done, we provide a similarity sketch between the case where $k=1$ and when $k$ is arbitrary.

\subsection{{Case 2($k\ge 2$)}}
In this section, we use the same analogy of covering the path using subpaths as in the previous case. Further, for the case where $k=1$, 
the center of all the subpaths were informed only through the single modulator vertex.
However, in the case where $k\geq 2$, the sub paths can belong to different vertices in the modulator depending on who informed the center vertex of the sub paths. To distinguish a sub path belonging to a vertex $m_i$ in the modulator from a sub path belonging to a vertex $m_j$ in the modulator, we use the notation - $SP(m_i)$ to mean all the sub paths that belong to $m_i$ (i.e. resulting from $m_i$ informing a vertex in the path $P$ which becomes the center vertex in the corresponding sub path). 

In this case, an information can be transmitted from modulator to path vertices or from a path to modulator vertices.
This is the crucial difference between the two cases (otherwise, the algorithm would have been ``almost" identical). This now forces us to consider some paths that are centered at a vertex and is used to cover a vertex in the path that will in turn transmit the information to a modulator vertex. This change forces us not to assume that all the paths are of the same length $2t-1$. However, note that there can be at most $k$ many such paths (possibly belonging to different vertices in the modulator). Even an exhaustive search over all the cases of placement of these ``type'' of paths would be bounded by $n^{f(k)}$, i.e. \textsf{XP}.


\noindent\textbf{Guessing all necessary information:} For each modulator vertex $m_i$, we guess a tuple describing how and when it first receives the information, as follows.
First, we guess the time $t_i^{\rightarrow} \in \{1, \dots, t\}$ at which $m_i$ is informed, which lies within the first $t$ steps. Second, we guess the source of this information: $m_i$ may be informed either by a vertex of the path $P$ or by another modulator vertex, and we guess which of the two cases occurs. Third, depending on this choice, we guess the informing vertex $v_i$ itself, so that $v_i \in V(P)$ in the first case and $v_i$ is a modulator vertex in the second. Fourth, we guess the vertex $vm_i$ in the modulator, whose subpath covers $v_i$. Fifth, we guess the vertex $u_i \in V(P)$ that lies at the center of the subpath associated with $vm_i$. Finally, we guess the time $t_i^{\leftarrow}$ at which $u_i$ is first informed by $vm_i$.

   Once this is done, consider the guess that corresponds to the optimal case. According to the guess, we assign a sub path belonging to the guessed modulator vertex $vm_i$ at the guessed vertex $u_i$ and transmit the information to either direction, finally reaching $v_i$ by the time $t_i^{\rightarrow}-1$ and transmitting the information to $m_i$ at the time $t_i^{\rightarrow}$. Now, by loosing the same number of sub paths (across all the $SP(m_i)'s$) to inform the modulator vertices, we have also covered a super set of vertices in the path $P$ when compared to the covering in the optimal case restricted to the paths centered at $u_i$'s. Now we are ready to start the $\DP$ procedure to cover the remaining vertices of the path $P$. 
   $$\DP(n_1,n_2,\cdots, n_k, p_i)=1$$

   If the vertices till $p_i$ (in $P$, the vertices from $p_1$ till $p_i$) can be covered by $n_j$ many sub paths of length $2t-1$ belonging to $SP(m_j) $  $ \forall j\in[k]$, where $n_j\leq |SP(m_j)|$-\{the paths used to inform the modulator vertices\} (the reduction in the amount is basically because, the reduction amount is exactly the many times we used the modulator vertex $m_j$ to inform something in the path when we guessed and the remaining amount is the number of times we can inform an un informed modulator vertex within the first $t$ time). The $\DP$ entry is 0 otherwise. 
   
   The $\DP$ entries are updated by checking if a ``certain'' number of paths belonging to ``certain'' many modulator vertices is enough top cover the path $P$ till a ``certain'' vertex, to achieve which, we essentially check (consult the previously stored information/$\DP$ entry) if it was possible to cover the path till a previous vertex using all but one available sub paths and then cover the excess with the additional path. We now provide the formal $\DP$ update step-
   $$\DP(n_1,n_2,\cdots, n_k, p_i)=1$$
   $$\iff$$

    \begin{enumerate}
    \item ( $p_i$ already covered by the guessing phase) the interval $[p_1,p_i]$ is already contained in the coverage produced by the (at most) $k$ sub-paths used in the initial step to inform the modulator vertices, using at most $n_j$ paths of $SP(m_j)$ for each $j$ from that step, \emph{or}

    \item There exists $j\in[k]$, $l<i$ and a sub-path in $SP(m_j)$ such that
    $$
    \DP(n_1,\dots,n_{j-1},\,n_j-1,\,n_{j+1},\dots,n_k,\;p_{i-l})=1,
    (l<i),$$
     such that placing the sub path in $SP(m_j)$ centered at some adjacent vertex  of $m_j$ (located between $p_{i-l-2t}$ and $p_{i}$) would result in covering the portion of $P$ from $p_1$ and $p_i$.
\end{enumerate}

For the runtime, since we are making $n^{5}$ many guesses for each vertex in the modulator and finally solving a $DP$ that has $n^k$ entry at each of the $|p|$ vertices (which requires iterating over $n$ choices of $l$ and $j$), we have the runtime bounded by $n^{\mathcal{O}(k)}$. Hence, we have an \textsf{XP} algorithm.

\begin{theorem}
    \textsc{Telephone Broadcasting} admits a 2 factor approximation when parameterized by the distance to path parameter $k$ with runtime $n^{\mathcal{O}(k)}$.
\end{theorem}

\section{An \textsf{FPT} Additive Approximation Algorithm Parameterized by (vertex) Distance to Clique}

We fix a clique modulator $M$ with $|M|=k$ and let $C=V\setminus M$, $|C|=n$.

\subsection{Structural observation}
Because $C$ is a clique, once $j\geq$ 1 of its vertices are informed, the remaining vertices of $C$ can all be informed within $\lceil \log_2(n-k)\rceil$ further rounds by doubling. Hence the \emph{only} nontrivial bottleneck is getting information \emph{into} and \emph{out of} the small side $M$. In particular:

\begin{observation}
\label{obs:few-actors}
In any optimal broadcast protocol, at most $k$ vertices of $C$ are ever
directly involved in informing a vertex of $M$ (at most one informer per
vertex of $M$), and at most $k$ vertices of $M$ are ever used to transmit
information back into $C$ (by informing a not-yet-informed neighbor in $C$).
\end{observation}

This observation is what lets us avoid branching over the whole clique: we
only need to \emph{guess} and \emph{order} a set of size $O(k)$ inside $C$,
and then use color coding to locate, among the $\binom{n}{O(k)}$ candidates,
a correctly ordered and positioned witness set, all in \textsf{FPT} time.

\subsection{High-level description of the algorithm}

We first treat the case where a single modulator vertex holds the initial
information; we then show how the general case reduces to this one. We
further split the algorithm into two regimes, based on the relation between
$n$ (the clique size) and $k$ (the modulator size):

\begin{itemize}
    \item \textbf{Case 1:} $k \cdot 2^k < n$
    \item \textbf{Case 2:} $k \cdot 2^k \geq n$
\end{itemize}

Part of the construction is common to both regimes, so we describe it first.

Consider the modulator set $M$. We guess a partition
$\mathcal{P} = \{P_1, P_2, \dots, P_{|\mathcal{P}|}\}$ of $M$ according to the
round in which each vertex is informed: two modulator vertices lie in the
same part $P_i$ iff they are informed in the same round.

We also guess the subset $S \subseteq M$ of modulator vertices that are
informed directly by a vertex of $C$. For every vertex of $M \setminus S$, we
additionally guess the specific modulator vertex that informs it (since such
vertices are informed from within $M$).

This completes the common part of the algorithm. We now describe each
regime.

\subsubsection{Algorithm for Case 1}

Here $k \cdot 2^k \leq n$, i.e. $\log k + k \leq \log n \leq T_{\text{opt}}$ (the optimum time). So there is ample
time to first inform some $k$ vertices of the clique and then use them to
inform all of $M$.

\begin{observation}\label{strtforward}
Once a single vertex of $C$ is informed, there is a polynomial-time procedure
that informs every vertex of $M$, together with at most $k$ further vertices of $C$.
\end{observation}

For each part $P_i$, fix a matching $M_i$ between $S \cap P_i$ and vertices
of $C$ that saturates $S \cap P_i$. Once every clique vertex used in one of
these matchings is informed, the modulator set can be informed in the next
$|\mathcal{P}|$ rounds by processing the parts $P_1, P_2, \dots$ in order,
one per round: within round $i$, the vertices of $S \cap P_i$ are informed
via the matching $M_i$ from their matched partners in $C$, and the remaining
vertices of $P_i$ are informed from the modulator vertices that were guessed
to inform them (Provided the initial guesses are correct). 

The full broadcasting procedure is:

\begin{enumerate}
    \item Inform two vertices of the clique in the first 2 rounds.
    \item The first of these two vertices runs the procedure above
          (Observation~\ref{strtforward}).
    \item The second vertex greedily propagates the information to all
          unmatched vertices of $C$ (taking time $O(\log n)$).
\end{enumerate}

Since the two branches run in parallel, the whole graph is broadcasted in time
$\max(\log n, \log k +k) + 2 \leq T_{\text{opt}} + 2$. This gives a $+2$
approximation in this regime. 

\subsubsection{Algorithm for Case 2}

Here $n \leq k \cdot 2^k$, so the time required to broadcast the whole graph satisfies $\log n + k \leq 2k + \log k = O(k)$. This means we can afford to guess the full partition of $M$ \emph{and} the exact round at which each part is informed.

Before describing the algorithm, note that every vertex of $C$ can be
assigned uniquely to a vertex of $M$: iteratively, a clique vertex inherits the modulator assignment of whichever vertex first informed it and if a vertex in the clique is informed by a modulator vertex, it is assigned to the same modulator vertex.This partitions $C$ into sub-cliques, one or more per modulator vertex (a single modulator vertex may seed several sub-cliques if it informs multiple clique vertices that later inform others clique vertices).

We guess the following:
\begin{enumerate}
    \item For each vertex of $M \setminus S$, we are guessing, which earlier-informed modulator vertex informs it.
    \item A partition of $S$ into subsets $S_1, \dots, S_d$, where each $S_i$
          is informed (in possibly different rounds) by a single vertex
          $s_i \in C$. Since two modulator vertices informed in the same
          round cannot share an informer, we require
          $|S_i \cap P_j| \leq 1$ for all $i, j$.
    \item Which of the vertices in $S_{\text{direct}} = \{s_i : i \in [d]\}$ are informed directly by a modulator vertex, rather than by another clique vertex. Call this subset $M_{S_{\text{direct}}} \subseteq S_{\text{direct}}$.  
    \item The round $t_{s_i}$ at which each $s_i$ is informed.
    \item For each $s_i \in M_{S_{\text{direct}}}$, guess the modulator vertex $m_i$ that informs it directly.
    \item For each $s_i \in S_{\text{direct}} \setminus M_{S_{\text{direct}}}$, guess which sub-clique it belongs to (two such vertices may lie in the same or different sub-cliques).
    \item For each of the (at most $k$) sub-cliques identified in step 6, guess the modulator vertex is associated with.
    \item For each such sub-clique, guess the round at which its first vertex received the information from that modulator vertex.
\end{enumerate}

The feasibility of a guess is not mentioned here, but informally, a guess is feasible/realizable if there are no two vertices in the graph that demand to be informed at the same time/round and more precisely if it is possible to inform the vertices at the exact round they are guessed to be informed (which can be checked in poly time for each guess). 

With these guesses fixed, we run a color-coding procedure:
\begin{enumerate}
    \item The number of vertices $s_i$ is at most $k$ (more precisely $d$),
    and the number of sub-cliques guessed in step 6 is at most $k$(more precisely, say $g\leq k$). So there are at most $d + g \leq 2k$ target vertices i.e. the $s_i$'s and the first-informed vertices of the guessed sub-cliques that we must correctly locate among the $n$ vertices of $C$. Order these target vertices arbitrarily as $ver_1, \dots, ver_{d+g}$.
    \item Randomly color the vertices of $C$ with $d + g$ colors.
    \item Given correct guesses, the probability that $ver_i$ receives color $i$ for every $i \in [d+g]$ is at least $(d+g)^{-(d+g)} \geq k^{-k}$.
\end{enumerate}

If the guesses and the coloring are both correct, we can, in polynomial
time, select one vertex from each relevant color class satisfying the
required adjacency conditions, e.g., if $ver_i$ plays the role of $s_i$, we
need a vertex of color $i$ adjacent to all of $S_i$ and to the appropriate
vertex of $M$; if $ver_i$ plays the role of a sub-clique's first informed
vertex, we need a vertex of color $i$ adjacent to the corresponding modulator vertex $m_j$; and so on.

Hence, conditional on correct guesses and correct coloring, we obtain a
broadcasting schedule that informs all of $M$ within time $T_{\text{opt}}$,
mimicking the optimal protocol:

\begin{observation}\label{case2}
If the optimal broadcast time is $T_{\text{opt}}$ and $n \leq k \cdot 2^k$,
there is an algorithm running in time $2^{O(k \log k)}$ that (possibly)
informs some vertices of $C$ together with all vertices of $M$, within
$T_{\text{opt}}$ rounds.
\end{observation}

Using Observation~\ref{case2}, we build the full algorithm exactly as in
Case 1:

\begin{enumerate}
    \item In 2 rounds, inform a vertex $v \in C$ not already informed by the
          procedure of Observation~\ref{case2}.
    \item Run the guessing-and-color-coding procedure above, which informs
          all vertices of $M$ (and some vertices of $C$) as described.
    \item Let $v$ propagate the information greedily through the remaining,
          not-yet-informed vertices of $C$.
\end{enumerate}

Since the two procedures run on disjoint vertex sets and hence in parallel,
the total number of rounds is $2 + \max(T_{\text{opt}}, \log n) \leq
T_{\text{opt}} + 2$. 

The color-coding step succeeds with probability at least $k^{-k}$
(conditional on correct guesses), so the algorithm can be derandomized using
a $k$-perfect hash family, yielding a deterministic \textsf{FPT} additive
approximation.

\paragraph{Runtime analysis.} The number of guesses is bounded by
$2^{O(k \log k)}$, and the derandomization step also runs in time
$2^{O(k \log k)}$. Hence the total runtime is
$2^{O(k \log k)} \cdot \text{poly}(n + k)$.Therefore, we have the theorem-

\begin{theorem}
    \textsc{Telephone Broadcasting} admits a +2 additive approximation when parameterized by the distance to clique parameter $k$ with the runtime $2^{\mathcal{O}(k\log(k))}$.
\end{theorem}

\section{Polynomial Time Algorithm on Polar graphs}\label{sec:polar}

\begin{lemma}\label{lem:largetime}
Given a clique of size $n$, where only one vertex has the information, all the vertices of the clique can be informed in time $\lceil \log n \rceil$.
\end{lemma}

\begin{proof}
Consider a clique on $n$ vertices. At each time step, every informed vertex can inform exactly one adjacent (uninformed) vertex. Let $I(t)$ denote the number of informed vertices by the end of step $t$. Initially, $I(0)=1$. Since the graph is a clique, every vertex is adjacent to all others. Hence, as long as there are enough uninformed vertices, all informed vertices can simultaneously inform distinct uninformed vertices. Therefore, when $I(t) \le n - I(t)$, we have $I(t+1) = 2I(t)$.

Once $I(t) > n - I(t)$, the number of remaining uninformed vertices is less than the number of informed vertices, so all remaining vertices can be informed in one additional step.

Hence, after $t$ steps, the total number of informed vertices is
$1 + 2 + 4 + \cdots + 2^{t-1} = 2^t - 1.$
To inform all $n$ vertices, it suffices that $2^t \geq n.$
This implies
$t \geq \log_2 n$. 

Since time is discrete, taking $t = \lceil \log_2 n \rceil$ is sufficient to ensure that all vertices in the clique are informed.
\end{proof}

\begin{lemma}\label{maxcliqueDP}
    Given a polar graph, with only one vertex $v$ in the independent side, there exists an algorithm, that in polynomial time decides the maximum number of cliques (on the clique side) that can be completely informed, when the total available number of steps is $\lceil \log(n) \rceil$.
\end{lemma}

\begin{proof}
    We give a polynomial run time Dynamic Programming algorithm to achieve the required number of max cliques, assuming the cliques are arranged in some particular order (we order them randomly, and then fix the ordering). We start with defining the $\DP$ entry.
    $$\DP (i,S\subseteq [\lceil \log(n) \rceil])$$
    This DP entry stores the maximum number of cliques, from the first $i$ cliques (in the fixed ordering) that can be completely informed by the end, assuming only at the time steps contained in $S$ can be used to transmit information from $v$ to one of the first $i$ cliques. We now go on to specify the exact $\DP$ update steps required to construct all the entries (iteratively).
    $$
    \mathrm{DP}\bigl(i,\, S \subseteq [\lceil \log n \rceil]\bigr)
    =
    \max \left\{
    \mathrm{DP}(i-1, S),\;
    \max_{\substack{S' \subseteq S \\ S' \text{ covers clique } i}}
    \mathrm{DP}\bigl(i-1,\, S \setminus S'\bigr)+1\right\}
    $$
    First, we clarify what ``covers" mean. Note that the vertices in the clique can be arranged in an order such that the vertices that are the neighbors of $v$ are not informed in a step by some other informed vertex in the clique while some non neighbor vertex is uninformed. Therefore, we assume that whenever an informed vertex in the clique transmits the information , it only transmits it to an vertex neighboring $v$ provided all the non adjacent clique vertices have been already informed. In this sense, we say that $S'\subseteq  [\lceil \log n \rceil]$ covers a clique adjacent to $v$ if transmitting information from $v$ to the uninformed clique vertices (neighboring $v$) at steps in $S'$ ensures that the entire clique is informed by the end.
    
    The proof of the DP step is straightforward if one notices that the two terms in the first max function corresponds to when a vertex in clique $i$ is not informed from $v$ at a time step in $S$ and otherwise respectively.
  
    As with general DP algorithms, we can backtrack to exactly know which of the cliques are informed and at what steps.
\end{proof}

\begin{theorem}
     Given a polar graph, with only one vertex $v$ in the independent side, there exists an algorithm, that correctly determines if all the vertices in the graph can be informed within time $t$, provided the information starts off with $v$.
\end{theorem}

\begin{proof}
    Note that if $t\leq \lceil \log(n) \rceil$, we can just run the same algorithm as presented in lemma~\ref{maxcliqueDP}.

    If $t > \lceil \log(n) \rceil$, from lemma~\ref{lem:largetime}, we know that if $v$ transmits the information to any of the clique at any of the steps from $1$ till $t-\lceil \log(n) \rceil$ (both inclusive), by the end of $t$ steps, that clique would be completely informed. Since all the cliques are adjacent to the vertex $v$, what remains is to determine the maximum number of cliques that can be informed in the last $\lceil \log(n) \rceil$ steps. Once we know this by running the DP algorithm of lemma~\ref{maxcliqueDP} , if the remaining number of cliques if more than $t-\lceil \log(n) \rceil$, we know that informing all the vertices is not possible in time $t$, however if remaining number of cliques is at most $t-\lceil \log(n) \rceil$, we know that its a yes instance by informing a vertex from these cliques in the first $t-\lceil \log(n) \rceil$ steps and in the remaining steps, proceeding with the DP algorithm of lemma~\ref{maxcliqueDP} (which we ran in the beginning).
\end{proof}

Therefore, we have,
\begin{theorem}
    \textsc{Telephone Broadcasting} admits a polynomial time algorithm running in time $\mathcal{O}({n^3})$ on polar graphs.
\end{theorem}
\bibliographystyle{alpha}
\bibliography{references}
\end{document}